\documentclass[11pt,a4paper]{article}

\usepackage{amsfonts,amsmath,amssymb,amsthm}
\usepackage{latexsym,amscd}
\usepackage{amsbsy}
\usepackage{CJK}
\usepackage{indentfirst,mathrsfs}
\usepackage{amsfonts,amsmath,amssymb,amsthm}
\usepackage{latexsym,amscd}
\usepackage{amsbsy}
\usepackage{color}
\usepackage{multirow}
\usepackage{makecell}
\usepackage{float}
\newtheorem{theorem}{Theorem}            

\numberwithin{equation}{section}

\newtheorem{lemma}{Lemma}

\newtheorem{proposition}{Proposition}

\newtheorem{definition}{Definition}

\newtheorem{remark}{Remark}

\newenvironment{proof1}{{\noindent\it Proof of Theorem \ref{thm-main}.}\,}{\hfill $\square$ \par}

\newcommand{\f}{\mathbb{F}_{2^n}}
\newcommand{\fm}{\mathbb{F}_{2^m}}
\newcommand{\fs}{\mathbb{F}_{2^n}\backslash\mathbb{F}_{2^m}}
\newcommand{\uc}{\mu_{2^m+1}}
\newcommand{\miu}{\mu_{2^m+1}\backslash\{1\}}

\begin{document}
\title{On differential properties of a class of Niho-type power functions}
\author{Zhexin Wang
\thanks{Z. Wang and N. Li are with the Hubei Key Laboratory of Applied Mathematics, School of Cyber Science and Technology, Hubei University, Wuhan, 430062, China, and also with the State Key Laboratory of Integrated Services Networks (Xidian University), Xi'an, 710071, China. Email: zhexin.wang@aliyun.com, nian.li@hubu.edu.cn}
, Sihem Mesnager
\thanks{S. Mesnager is with the Department of Mathematics, University of Paris VIII, Saint-Denis, Paris, France, and also with the Laboratory of Analysis, Geometry, and Applications (LAGA), University Sorbonne Paris Nord CNRS, UMR 7539, Villetaneuse, France, and also with the Telecom Paris, Palaiseau, France. Email: smesnager@univ-paris8.fr}
, Nian Li and Xiangyong Zeng
\thanks{X. Zeng is with the Hubei Key Laboratory of Applied Mathematics, Faculty of Mathematics and Statistics, Hubei University, Wuhan, 430062, China. Email: xiangyongzeng@aliyun.com}
}
\date{\today}
\maketitle
\begin{quote}
  {\small {\bf Abstract:}

    This paper deals with Niho functions which are one of the most important classes of functions thanks to their close connections with a wide variety of objects from mathematics, such as spreads and oval polynomials or from applied areas, such as symmetric cryptography, coding theory and sequences.
    In this paper, we investigate specifically the $c$-differential uniformity of the power function $F(x)=x^{s(2^m-1)+1}$ over the finite field $\f$, where $n=2m$, $m$ is odd and $s=(2^k+1)^{-1}$ is the multiplicative inverse of $2^k+1$ modulo $2^m+1$, and show that the $c$-differential uniformity of $F(x)$ is $2^{\gcd(k,m)}+1$ by carrying out some subtle manipulation of certain equations over $\f$. Notably, $F(x)$ has a very low $c$-differential uniformity equals $3$ when $k$ and $m$ are coprime.
  }

 {\small {\bf Keywords:} Power function, $c$-differential uniformity, Niho exponent.
  }

\end{quote}

\section{Introduction}\label{intro}

A {\em Dillon exponent} with respect to the finite field $\f$ is $s(2^m-1)$, where $n=2m$ and $s$ is coprime with $2^m+1$. Moreover, a positive integer $d$ (always understood modulo $2^n-1$) is said to be a {\em Niho exponent}, and $x^d$ is a {\em Niho power function}, if the restriction of $x^d$ to $\fm$ is linear or in other words $d\equiv 2^j\pmod {2^m-1}$ for some $j<n$.
Niho functions and, more significantly, those which are power functions, for their low implementation cost in hardware, have attracted much attention explored from their (algebraic) property initially inherited from applied domains, such as their nonlinearity or their differential property when the properties are optimal giving rise, respectively to maximally nonlinear functions (or bent functions) and low-differential uniformity or their generalization, specifically low $c$-differential uniform functions, that have direct impacts on combinatorial designs. A recent relationship between linear cyclic codes and low differentially uniform functions has been pointed out in \cite{SM.MS2022}.
In this article, we are interested in the differential properties of a class of Niho-type power functions. Initially, such a differential uniformity property \cite{KN1994} was introduced for a cryptographic purpose directly related to differential cryptanalysis, one of the most fundamental cryptanalytic approaches targeting symmetric key primitives and the first statistical attack for breaking iterated block ciphers \cite{EB.AS1991}. Next, it has been extended in recent years to $c$-differential uniformity \cite{PE.PF2020} for a specific non-zero element $c$ living the finite field. More importantly, an interesting connection between $c$-differential uniformity and combinatorial designs has been highlighted in \cite{NA.TK2022} by showing that the graph of a perfect $c$-nonlinear function (an optimal function with respect to the $c$-differential uniformity) is a set of differences in a quasigroup. Difference sets give rise to symmetric designs, which are known to build optimal self-complementary codes. Some types of designs also have application implications, such as secret sharing and visual cryptography.

Below we recall the basic definitions of the framework of the $c$-differential uniformity for functions defined over any finite field.
\begin{definition}{\rm(\cite{PE.PF2020})}\label{def-c-diff}	
    Let $\mathbb{F}_{p^n}$ denote the finite field with $p^n$ elements, where $p$ is a prime number, and $n$ is a positive integer. For a function $F$ over $\mathbb{F}_{p^n}$, and $c\in\mathbb{F}_{p^n}$, define the multiplicative $c$-derivative of $F$ with respect to $a\in\mathbb{F}_{p^n}$ as
    $$_cD_aF(x) = F(x+a)-cF(x)$$
    for each $x\in\mathbb{F}_{p^n}$. Denote $_c\Delta_F(a,b)=\#\{ x\in\mathbb{F}_{p^n}: {_c}D_aF(x) = b \}$ for $b\in\mathbb{F}_{p^n}$, and call $_c\Delta_F = \max\{ _c\Delta_F(a,b): a,b\in \mathbb{F}_{p^n}\ and\ a\ne 0\ if\ c=1 \}$ the $c$-differential uniformity of $F$.
\end{definition}

Finding functions, particularly permutations, with good differential properties has recently received much attention due to their possible applications in cryptography and combinatorial design. In this paper, we investigate the $c$-differential uniformity of a class of power functions over finite fields of even characteristics with low $c$-differential uniformity. The latter has been widely studied in recent years due to their simple algebraic form and lower implementation costs in hardware. The reader is referred to \cite{SM.CR2021, PE.PF2020, ZT.NL2023, XW.DZ2022, HY.KZ2022} and the references therein. It is typically challenging to obtain power functions with low $c$-differential uniformity.
To the best of our knowledge, we summarize the known power functions over finite fields of even characteristics with $_c\Delta_F \leq 3$ in Table \ref{tab}.

\begin{table}\label{tab}
\caption{Power functions $F(x) = x^d$ over $\f$ with $_c\Delta_F \leq 3$}
\begin{center}
\begin{tabular}{|l|p{5.6cm}|c|c|}
\hline
$d$  &  Conditions on $n$ and $c$  &  $_c\Delta_F$  &  Reference \\
\hline\hline
$2^k+1$  &  $n\geq3$, $c\in\mathbb{F}_{2^{\gcd(k,n)}}\backslash\{1\}$  &  $1$  &  \cite{SM.CR2021} \\ \hline
$2^n-2$  &  $c\ne 0$, ${\rm Tr}(c) = {\rm Tr}(1/c)=1$  &  $2$  &  \cite{PE.PF2020} \\ \hline
$2^{3m}+2^{2m}+2^m-1$ & $n=4m$, $c\in\mu_{2^{2m}+1}\backslash\{1\}$  &  $2$  &  \cite{ZT.NL2023} \\ \hline
$2^n-2$  &  $c\ne 0$, ${\rm Tr}(c)=0$ or ${\rm Tr}(1/c) = 0$  &  $3$  &  \cite{PE.PF2020} \\ \hline
$2^k+1$  &  $\gcd(k,n)=1$, $c\in\mathbb{F}_{2^n}\backslash\mathbb{F}_2$  &  $3$  &  \cite{SM.CR2021} \\ \hline
$(2^{n+1}-1)/3$  &  $2^n\equiv 2\pmod 3$, $c\ne 1$  &  $\leq 3$  &  \cite{SM.CR2021} \\ \hline
$(2^m-1)/(2^k+1)+1$ & $n=2m$, $m$ odd, $c^{2^m+1}=1$, $c\ne 1$  &  $3$  &  Theorem \ref{thm-main} \\ \hline
\end{tabular}
\end{center}
\end{table}

Very recently, Xie, Mesnager, Li et al. studied the differential property of the power function $F(x)=x^{s(2^m-1)+1}$ over the finite field $\f$ and determined its differential spectrum \cite{XX.SM2022}, where $n=2m$, $s$ is the multiplicative inverse of $2^k+1$ modulo $2^m+1$ and $\gcd(k,m)=1$. This inspired us to investigate the $c$-differential uniformity of this power function $F(x)$. Consequently, by making the most use of certain techniques in solving equations over finite fields, different from the ones used in \cite{XX.SM2022}, we show that for $c\in\f\backslash\{1\}$ satisfying $c^{2^m+1}=1$ the $c$-differential uniformity of $F(x)$ is equal to $2^{\gcd(k,m)}+1$ when $m$ is odd with $ \gcd(2^k+1,2^m+1)=1$. In particular, we can obtain an infinite class of power functions over the finite field $\f$ with $c$-differential uniformity $3$ if $\gcd(k,m)=1$. By comparing the algebraic degrees and the values of $c$, it can be readily verified that the functions with $c$-differential uniformity $3$ in this paper are not equivalent to the known ones in Table \ref{tab} up to the equivalence relation introduced in \cite{SH.MP2021}.

The remainder of this paper is organized as follows. Section \ref{prel} introduces some notions and auxiliary results. Section \ref{diff} presents the main results and their proofs. Section \ref{conc} concludes this study.

\section{Preliminaries}\label{prel}

Let $n=2m$ and $m$ is odd. Let $\f$ be the finite field with $2^n$ elements and $\f^*=\f\backslash\{0\}$. Denote the unit circle of $\f$ by
$$\mu_{2^m+1}=\{ x\in\f: x^{2^m+1}=1 \}.$$

A helpful characterization for the $c$-differential uniformity of power functions is given below.

\begin{lemma}{\rm(\cite[Lemma 1]{SM.CR2021})}\label{lem-c-power}
    Let $F(x)=x^d$ be a power function over $\f$. Then
    $$_c\Delta_F = \max\{ \{_c\Delta_F(1,b): b \in \f\} \cup \{\gcd(d,2^n-1)\} \}.$$
    More precisely,
    $$_c\Delta_F(0,b) = \begin{cases}
		1, &  {\rm if}\ b = 0, \\
		\gcd(d,2^n-1), &  {\rm if}\, \frac{b}{1-c}\in\f^*\; {\rm is}\;  {\rm a}\;  \mbox{d-th}\,{\rm power}, \\
		0, &  {\rm otherwise}. \end{cases}$$
\end{lemma}
The following lemmas will be used to prove our main results in the sequel.

\begin{lemma}{\rm(\cite{HD.PF2006})}\label{lem-quad}
    Let $n,m,r$ be positive integers, and $n=2m$. The quadratic polynomial
    $$Q(x) = x^{2^r+1}+ax^{2^r}+bx+c, \quad a,b,c\in\mathbb{F}_{2^n}$$
    has either $0,1,2$, or $2^{r_0}+1$ roots in $\mathbb{F}_{2^n}$, where $r_0=\gcd(n,r)$. Moreover, if $Q$ has three distinct roots $x_0,x_1,x_2\in\mu_{2^m+1}$, then $Q$ has $2^{r_1}+1$ distinct roots in $\mu_{2^m+1}$ where $r_1=\gcd(m,r_0)$.
\end{lemma}

\begin{lemma}\label{lem-map}
    Let $n,m$ be positive integers and $n=2m$. Define the set
    $$\mathcal{D} = \{ (u,v): u,v\in\miu,\ and\ u\ne v \},$$
    and the mapping $\varphi$ from $\mathcal{D}$ to $\fs$ as $\varphi: (u,v)\mapsto\frac{u^2(1+v^2)}{u^2+v^2}$, where $\mu_{2^m+1}$ is the unit circle of $\f$. Then $\varphi$ is a bijection.
\end{lemma}
\begin{proof}
For any $(u_1,v_1)\ne (u_2,v_2)$, define $\alpha_i = \frac{u_i(1+v_i^2)}{u_i^2+v_i^2}$, $i\in\{1,2\}$. Then one gets $\alpha_i \in\fm$ since $u_i, v_i\in\mu_{2^m+1}$. Assume that $\varphi(u_1,v_1) = \varphi(u_2,v_2)$, i.e.,
$$\frac{u_1^2(1+v_1^2)}{u_1^2+v_1^2} = \frac{u_2^2(1+v_2^2)}{u_2^2+v_2^2},$$
one then has
$$\frac{u_1}{u_2} = \frac{u_2(1+v_2^2)}{u_2^2+v_2^2}\cdot\frac{u_1^2+v_1^2}{u_1(1+v_1^2)} = \frac{\alpha_2}{\alpha_1}.$$
Since $u_1/u_2\in \mu_{2^m+1}$, $\alpha_2/\alpha_1\in\fm$ and $\mu_{2^m+1}\cap\fm=\{1\}$, hence we have $u_1=u_2$ and $\alpha_1=\alpha_2$. This leads to
$$\frac{u_1(1+v_1^2)}{u_1^2+v_1^2} = \frac{u_1(1+v_2^2)}{u_1^2+v_2^2},$$
i.e., $(1+u_1)^2(v_1+v_2)^2=0$, which implies $v_1=v_2$ as $u_1\ne 1$. Therefore, $(u_1, v_1)=(u_2, v_2)$, a contradiction with the assumption. That is, $\varphi$ is an injection. Then the desired result follows from the fact $\#\mathcal{D}=2^n-2^m$=$\#(\fs)$. This completes the proof.
\end{proof}

\section{The $c$-differential uniformity of $x^{s(2^m-1)+1}$ for odd $m$}\label{diff}

In this section, we investigate the $c$-differential uniformity of the power function $F(x)=x^{s(2^m-1)+1}$ over $\f$, where $n=2m$, $m$ is odd, $\gcd(2^k+1,2^m+1)=1$ and $s=(2^k+1)^{-1}$ is the multiplicative inverse of $2^k+1$ modulo $2^m+1$.

Our main result is stated below.

\begin{theorem}\label{thm-main}
    Let $F(x)=x^{s(2^m-1)+1}$ be a power function over $\f$, where $n=2m$, $m$ is odd and $s=(2^k+1)^{-1}$ with $\gcd(2^k+1,2^m+1)=1$. Then the $c$-differential uniformity of $F(x)$ is $_c\Delta_F=2^{\gcd(k,m)}+1$ for $c\in\miu$. In particular, $_c\Delta_F=3 $ if $\gcd(k,m)=1$.
\end{theorem}

\begin{remark}\label{rem-c}
    Computer experiments indicate that the $c$-differential uniformity of $F(x)$ in Theorem \ref{thm-main} varies with $m$ and seems not to be good when $c\notin\mu_{2^m+1}$.
\end{remark}

\begin{remark}\label{rem-m}
   Computer experiments show that the $c$-differential uniformity of $F(x)$ for even $m$ is either $2$ or $2^{\gcd(2k,m)}+1$ for $c\in\miu$. Our data indicate that for even $m$, the resulting almost perfect $c$-nonlinear functions are equivalent to known ones.
\end{remark}

To prove Theorem \ref{thm-main}, according to Definition \ref{def-c-diff} and Lemma \ref{lem-c-power}, we first discuss the number of solutions of
\begin{equation}\label{eq-c-diff}
	(x+1)^{s(2^m-1)+1}+cx^{s(2^m-1)+1} = b.
\end{equation}

Denote the conjugate of $b$ by $\overline{b}=b^{2^m}$ for any $b\in\f$. First, for the special case $x\in\fm$, we have

\begin{proposition}\label{prop-fm}
    Let $c\in\miu$ and $b\in\mathbb{F}_{2^n}$. Then \eqref{eq-c-diff} has at most one solution in $\fm$. Further, it is solvable in $\fm$ if $c(\overline{b}+1)+b+1=0$ and the solution is $(b+1)/(c+1)$, where $\overline{b} = b^{2^m}$ is the conjugate of $b\in\f$.
\end{proposition}
\begin{proof}
Observe that $x=0$ (resp. $x=1$) is a solution of $\eqref{eq-c-diff}$ when $b=1$ (resp. $b=c$) and \eqref{eq-c-diff} can be reduced to $(c+1)x=b+1$ if $x\in\fm$. This implies that $\eqref{eq-c-diff}$ has at most one solution in $\fm$ for $c\in\miu$ and $b\in\f$. Moreover, the solution $x=(b+1)/(c+1)$ belongs to $\fm$ only if
$$(\frac{b+1}{c+1})^{2^m} = \frac{c(\overline{b}+1)}{c+1} = \frac{b+1}{c+1},$$
which leads to $c(\overline{b}+1)+b+1=0$. This completes the proof.
\end{proof}

Next, we consider the solutions of \eqref{eq-c-diff} in $\fs$. Using the polar representation of elements in $\f$, for any $x\in\fs$, assume that $x=\alpha u, x+1=\beta v$, where $\alpha,\beta\in\fm^*$ and $u,v\in\miu$. We then claim that every $x\in\fs$ can be uniquely determined by $(u,v)$. Note that $\alpha u+\beta v=1$. This together with the fact $u,v\in\miu$ gives
\begin{equation}\label{eqana-1}
	\begin{cases}
		\alpha u+\beta v = 1, \\
		\alpha u^{-1}+\beta v^{-1} = 1,
	\end{cases}
\end{equation}
which implies that
$$\begin{cases}
	\alpha(u^2+v^2) = u(v^2+1), \\
	\beta(u^2+v^2) = (u^2+1)v
\end{cases}$$
and $u\ne v$. Thus, we have
\begin{equation}\label{eqana-2}
	\alpha = \frac{u(1+v^2)}{u^2+v^2}, \\
	\beta = \frac{v(1+u^2)}{u^2+v^2},
\end{equation}
which shows that $\alpha,\beta\in\fm^*$ can be determined by $(u,v)$, and
$$x = \alpha u = \frac{u^2(1+v^2)}{u^2+v^2}.$$
Let $\mathcal{D} = \{ (u,v): u,v\in\miu,\ u\ne v \}$, according to Lemma \ref{lem-map}, there exists a unique $(u,v)\in\mathcal{D}$ such that $x=\frac{u^2(1+v^2)}{u^2+v^2}$ for each $x\in\fs$, that is, $x$ can be uniquely determined by $(u,v)$ for any $x\in\fs$.

For each $x\in\fs$, by $x=\alpha u$, $x+1=\beta v$, \eqref{eq-c-diff} can be written as
\begin{equation}\label{eqana-3}
	\begin{cases}
		\beta v^{1-2s}+c\alpha u^{1-2s} = b, \\
		\alpha u+\beta v = 1.		
	\end{cases}
\end{equation}
Multiplying $u^{2s}v^{2s}$ on both sides of the first equation of \eqref{eqana-3} and substituting $\alpha u=\beta v+1$ (resp. $\beta v=\alpha u+1$)
into it gives
$$\begin{cases}
	\alpha u(u^{2s}+cv^{2s}) = (bv^{2s}+1)u^{2s}, \\
	\beta v(u^{2s}+cv^{2s}) = (bu^{2s}+c)v^{2s}.
\end{cases}$$
Note that $u^{2s}$ is a one-to-one mapping over $\mu_{2^m+1}$ due to $\gcd(2s,2^m+1)=1$. For convenience, let $y=u^{2s}$ and $z=v^{2s}$, then $y\ne z$, $y,z\ne 1$, and $(y,z)$ can be uniquely determined by $(u,v)$. The above system of equations becomes
\begin{equation}\label{eqana-4}
	\begin{cases}
		\alpha u(y+cz) = y(bz+1), \\
		\beta v(y+cz) = (by+c)z,
	\end{cases}
\end{equation}
where $u\ne v$, $y\ne z$, and $u,v,y,z\ne 1$.

We then discuss the solutions of \eqref{eq-c-diff} in $\fs$ by \eqref{eqana-4}.

\begin{proposition}\label{prop-fs-1}
   Let $c\in\miu$ and $b\in\f$. With the notation as above, define a system of equations on $(y,z)$ with $y\ne z$ and $y,z\ne 1$ as follows:
	\begin{equation}\label{eqprop-1}
		\begin{cases}
			c(z+\overline{b})y^{2^k}+bz+1 = 0, \\
			(y+c\overline{b})z^{2^k}+by+c = 0.
		\end{cases}
	\end{equation}
    Then the solutions of \eqref{eq-c-diff} in $\fs$ can be expressed as
    $$\begin{aligned}
        & \{ \frac{b^{-(2^k+1})+1}{c^{-(2^k+1)}+1}: y+cz=0,\ \eqref{eqprop-1}\ holds \} \cup \{ \frac{y(bz+1)}{y+cz}: y+cz\ne 0,\ \eqref{eqprop-1}\ holds \}.
    \end{aligned}$$
\end{proposition}
\begin{proof}
We shall slip the discussions into the following two cases according to \eqref{eqana-4}.

\textbf{Case 1:} $y+cz=0$. If this case happens, then by \eqref{eqana-4} one immediately has $y=cb^{-1}$ and $z=b^{-1}$, which implies $b\in\mu_{2^m+1}\backslash\{1,c\}$ since $y,z\ne 1$. Since $y=u^{2s}$ and $z=v^{2s}$, where $s=(2^k+1)^{-1}$, then by $y+cz=0$, we have
$$u^2=y^{2^k+1}=c^{2^k+1}b^{-(2^k+1)},\ v^2=z^{2^k+1}=b^{-(2^k+1)}.$$
Thus, the pair $(u,v)$ is uniquely determined by $(c,b)$ for any fixed $c\in\miu$ and $b\in\mu_{2^m+1}\backslash\{1,c\}$. Further, by Lemma \ref{lem-map}, we have
$x=\frac{u^2(1+v^2)}{u^2+v^2}=\frac{b^{-(2^k+1})+1}{c^{-(2^k+1)}+1}\in\fs$, and consequently,
$$x^{2^m-1}=\frac{b^{2^k+1}}{c^{2^k+1}},\ x+1=\frac{c^{-(2^k+1)}+b^{-(2^k+1)}}{c^{-(2^k+1)}+1},\ (x+1)^{2^m-1}=b^{2^k+1}.$$
This leads to
$$(x+1)^{s(2^m-1)+1}+cx^{s(2^m-1)+1}=b\frac{c^{-(2^k+1)}+b^{-(2^k+1)}}{c^{-(2^k+1)}+1}+b\frac{b^{-{2^k+1}}+1}{c^{-{2^k+1}}+1}=b.$$
Again by Lemma \ref{lem-map} one can conclude that $x=\frac{b^{-(2^k+1})+1}{c^{-(2^k+1)}+1}$ is the unique solution of \eqref{eq-c-diff} in $\fs$ if $y+cz=0$ for any fixed $c\in\miu$ and $b\in\mu_{2^m+1}\backslash\{1,c\}$.

\textbf{Case 2:} $y+cz\ne 0.$
For this case, by \eqref{eqana-4}, we have
\begin{equation}\label{eqprop-2}
	\alpha = \frac{y(bz+1)}{u(y+cz)},
	\beta = \frac{(by+c)z}{v(y+cz)}.
\end{equation}
According to $\alpha=\alpha^{2^m}$ and $\beta=\beta^{2^m}$, we obtain \eqref{eqprop-1} and then $x=\alpha u=\frac{y(bz+1)}{y+cz}\in\fs$ since $u\ne 1$. This together with the fact $u\ne 1$ indicates that the solutions of \eqref{eq-c-diff} in $\fs$ is of the form $x=\alpha u=\frac{y(bz+1)}{y+cz}$, where $y\ne z$, $y,z\in\miu$ and satisfy \eqref{eqprop-1}.

In the following, we show that $x=\frac{y(bz+1)}{y+cz}$ is a solution of \eqref{eq-c-diff} in $\fs$ for any distinct $y$ and $z$ satisfying $y,z\in\miu$ and \eqref{eqprop-1}. Observe that $z+\overline{b}\ne 0$ and $y+c\overline{b}\ne 0$ if $b\not\in\uc$. Thus by \eqref{eqprop-1} we have
$$y^{2^k} = \frac{bz+1}{c(z+\overline{b})},\ y = \frac{c(\overline{b}z^{2^k}+1)}{z^{2^k}+b},\
z^{2^k} = \frac{by+c}{y+c\overline{b}},\ z = \frac{c\overline{b}y^{2^k}+1}{cy^{2^k}+b}.$$
This together with $x=\frac{y(bz+1)}{y+cz}$ implies that
$$x^{2^m-1} = y^{-(2^k+1)},\ x+1 = \frac{(by+c)z}{y+cz},\ (x+1)^{2^m-1} = z^{-(2^k+1)},$$
which indicates that $x\in\fs$ since $x^{2^m-1}=y^{-(2^k+1)}\ne 1$ and $x$ is a solution of \eqref{eq-c-diff} due to
$$(x+1)^{s(2^m-1)+1}+cx^{s(2^m-1)+1} = \frac{by+c}{y+cz}+c\frac{bz+1}{y+cz} = b.$$
For the remaining case $b\in\uc$, \eqref{eqprop-1} is reduced to
\begin{equation*}
	\begin{cases}
		(cb^{-1}y^{2^k}+1)(bz+1) = 0, \\
		(y+cb^{-1})(z^{2^k}+b) = 0.
	\end{cases}
\end{equation*}
If $bz+1=0$, i.e., $z=b^{-1}$, then $z^{2^k}+b=(1+b^{2^k+1})/b^{2^k}=0$ holds if and only if $b=1$ since $\gcd(2^k+1,2^m+1)=1$. This contradicts with $z\ne 1$. Thus we have $y+cb^{-1}=0$ and then $y+cz=0$, a contradiction. Hence, $bz+1\ne0$ and consequently, we obtain $cb^{-1}y^{2^k}+1=0$. Similarly, we can show that $y+cb^{-1}\ne 0$ and $z^{2^k}+b=0$. Therefore, we arrive at $y^{2^k}=c^{-1}b$ and $z^{2^k}=b$, where $b\notin\{1,c\}$ since $y,z\ne 1$. Then we can similarly prove that $x=\frac{y(bz+1)}{y+cz}$ belongs to $\fs$ and it is a solution of \eqref{eq-c-diff}. This completes the proof.
\end{proof}

In what follows, we give the proof of Theorem \ref{thm-main}.

\begin{proof1}
Since $\gcd(2^k+1,2^m+1)=1$ if and only if one of $\frac{k}{\gcd(k,m)}$ and $\frac{m}{\gcd(k,m)}$ is even, we have $k$ is even due to $m$ is odd. Then, according to Definition \ref{def-c-diff} and Lemma \ref{lem-c-power}, for the power function $F(x)=x^{s(2^m-1)+1}$, we have
$\max\{_c\Delta_F(0,b)\}=\gcd(s(2^m-1)+1,2^n-1)=\gcd(1-2s,2^m+1)=\gcd(2^k-1,2^m+1)=2^{\gcd(k,m)}+1$ since $\gcd(2^k+1,2^m+1)=1$ and $s=(2^k+1)^{-1}$.

Note that ${_c}\Delta_F(a,b)={_c}\Delta_F(1,\frac{b}{a^{s(2^m-1)+1}})$ for any $a\in\f^*$ and $b\in\f$. Thus, to determine the $c$-differential uniformity of $F(x)$, we need to count the number of solutions of \eqref{eq-c-diff} for $b\in\f$ and $c\in\miu$. For simplicity, define
$$\mathcal{V} = \{ x\in\f: x^{\frac{2^m+1}{2^{\gcd(k,m)}+1}} = 1 \}.$$

\textbf{Case 1:} $b=0.$
In this case, we have
$$c(\overline{b}+1)+b+1 = c+1 \ne 0$$
as $c\ne 1$. Thus \eqref{eq-c-diff} has no solution in $\fm$ by Proposition \ref{prop-fm}. Note that $x\notin\fm$ and $0=b\notin\uc$. Then, from the proof of Proposition \ref{prop-fs-1}, \eqref{eq-c-diff} becomes \eqref{eqprop-1}, which is reduced to
\begin{equation}\label{eqthm-1}
	\begin{cases}
		czy^{2^k}+1 = 0, \\
		yz^{2^k}+c = 0
	\end{cases}
\end{equation}
for $b=0$. This leads to $y=z^{-1}$ and $z^{2^k-1}=c$. Thus \eqref{eqthm-1} has $2^{\gcd(k,m)}+1$ solutions if $c\in\mathcal{V}$ and no solution otherwise since $\gcd(2^k-1,2^m+1)=2^{\gcd(k,m)}+1$. Hence, for $b=0$, by Proposition \ref{prop-fs-1}, \eqref{eq-c-diff} has at most $2^{\gcd(k,m)}+1$ solutions if $c\in\mathcal{V}$ and none solution otherwise, i.e., $_c\Delta_F(1,0)\leq 2^{\gcd(k,m)}+1$ if $c\in\mathcal{V}$ and $_c\Delta_F(1,0)=0$ otherwise.

\textbf{Case 2:} $b\in\mu_{2^m+1}$. In this case we have $\overline{b}=b^{-1}$ and
$$c(\overline{b}+1)+b+1 = (cb^{-1}+1)(b+1),$$
which implies that $c(\overline{b}+1)+b+1=0$ if and only if $b\in\{1,c\}$.

\textbf{Case 2.1:} $b=1.$
According to Proposition \ref{prop-fm}, $x=(b+1)/(c+1)=0\in\fm$ is a solution of \eqref{eq-c-diff}. On the other hand, for $x\in\fs$, from the proof of Proposition \ref{prop-fs-1} we have $y+cz\ne 0$ as $b=1$, thus \eqref{eqprop-1} is reduced to
\begin{equation}\label{eqthm-3}
	\begin{cases}
		(cy^{2^k}+1)(z+1) = 0, \\
		(y+c)(z^{2^k}+1) = 0.
	\end{cases}
\end{equation}
Since $y,z\ne 1$, then we have $y=c$ and $cy^{2^k}=1$, which leads to $y=c=1$, a contradiction. Thus \eqref{eqthm-3} has no solution and consequently \eqref{eq-c-diff} holds for $b=1$ if and only if $x=0$, i.e., ${_c}\Delta_F(1,1)=1$.

\textbf{Case 2.2:} $b=c$. Similar to Case 2.1, we can show that \eqref{eq-c-diff} has exactly one solution $x=1$ when $b=c$, i.e., ${_c}\Delta_F(1,c)=1$.

\textbf{Case 2.3:} $b\in\mu_{2^m+1}\backslash\{1,c\}.$ If this case occurs, then one obtains $c(\overline{b}+1)+b+1\ne 0$ and \eqref{eq-c-diff} has no solution in $\fm$ by Proposition \ref{prop-fm}. For $x\in\fs$, according to the proof of Proposition \ref{prop-fs-1}, one can conclude that $x= \frac{b^{-(2^k+1)}+1}{c^{-(2^k+1)}+1}$ and $x=\frac{y(bz+1)}{y+cz}$ are solutions of \eqref{eq-c-diff} when $b\in\mu_{2^m+1}\backslash\{1,c\}$, where $y$ and $z$ are uniquely determined by $y^{2^k}=c^{-1}b$ and $z^{2^k}=b$. Therefore \eqref{eq-c-diff} has at most two solutions in this case, i.e., ${_c}\Delta_F(1,b)\leq 2$ for $b\in\mu_{2^m+1}\backslash\{1,c\}$.

\textbf{Case 3:} $b\in\f^*\backslash\mu_{2^m+1}.$ If $c(\overline{b}+1)+b+1=0$, i.e. $c=(b+1)/(\overline{b}+1)$, then \eqref{eq-c-diff} has one solution in $\fm$. If $c(\overline{b}+1)+b+1 \ne 0$, then \eqref{eq-c-diff} has no solution in $\fm$ due to Proposition \ref{prop-fm}. For $x\in\fs$, according to the proof of Proposition \ref{prop-fs-1}, \eqref{eq-c-diff} becomes \eqref{eqprop-1} with $y+cz\ne 0$ when $b\in\f^*\backslash\mu_{2^m+1}$. Recall from \eqref{eqprop-1} that
$$\begin{cases}
	y^{2^k} = \frac{bz+1}{c(z+\overline{b})}, \\
	y = \frac{c(\overline{b}z^{2^k}+1)}{z^{2^k}+b}.
\end{cases}$$
Taking $2^k$-th power on both sides of the second equation and adding to the first one, we obtain the following equation about $z$ for $z\ne 1$,
\begin{equation*}\label{eqthm-6}
    (c^{2^k+1}\overline{b}^{2^k}+b)z^{2^{2k}+1}+(c^{2^k+1}\overline{b}^{2^k+1}+1)z^{2^{2k}}+(c^{2^k+1}+b^{2^k+1})z+c^{2^k+1}\overline{b}+b^{2^k} = 0.
\end{equation*}
Note that $\gcd(2^{m+k}-1,2^m+1)=1$ and $b\ne 1$. This indicates that $c^{2^k+1}\overline{b}^{2^k}+b=b(c^{2^k+1}b^{2^{m+k}-1}+1)\ne 0$. Then the above equation has at most $2^{\gcd(k,m)}+1$ solutions in $\mu_{2^m+1}$ from Lemma \ref{lem-quad}. Observe that $z=1$ is a solution of the above equation if $c=(b+1)/(\overline{b}+1)$. Thus the above equation has at most $2^{\gcd(k,m)}$ solutions in $\miu$ if $c=(b+1)/(\overline{b}+1)$, and at most $2^{\gcd(k,m)}+1$ solutions in $\miu$ if $c\ne (b+1)/(\overline{b}+1)$. Consequently, \eqref{eq-c-diff} has at most $2^{\gcd(k,m)}+1$ solutions in $\f$ for $b\in \f^*\backslash\mu_{2^m+1}$, i.e., ${_c}\Delta_F(1,b)\leq 2^{\gcd(k,m)}+1$ for $b\in\f^*\backslash\mu_{2^m+1}$.

Combining the above cases, we have $\max\{_c\Delta_F(1,b):b\in\f\}=2^{\gcd(k,m)}+1$ for all $b\in\f$. In particular, the $c$-differential uniformity of $F(x)$ is $3$ if $\gcd(k,m)=1$. This completes the proof.
\end{proof1}

\section{Conclusion}\label{conc}

In this paper, we investigated the $c$-differential uniformity of the power function of type Niho. More specifically, we proved that the $c$-differential uniformity of the power function $x^{s(2^m-1)+1}$ over $\f$ is $2^{\gcd(k,m)}+1$, where $m$ is odd, $\gcd(2^k+1,2^m+1)=1$, $s=(2^k+1)^{-1}$ and $1\ne c\in\f$ satisfying $c^{2^m+1}=1$. We obtained a class of power functions with $c$-differential uniformity $3$ when $\gcd(k,m)=1$, which is inequivalent to the known ones by comparing the algebraic degrees and the values of $c$. Although the notion of $c$-differential uniformity is originally a legacy of concepts of classical cryptographic interest and introduced for cryptographic motivation, our main interest in this article is to study it for families of important power functions aimed at deepening the results and discovering more in the sense of coding theory and for combinatorial interests.

\section*{Declaration of competing interest}

The authors declare that they have no known competing financial interests or personal relationships that could have appeared to influence the work reported in this paper.

\section*{Acknowledgments}

This work was supported by the National Key Research and Development Program of China (No. 2021YFA1000600), the National Natural Science Foundation of China (No. 62072162), the Natural Science Foundation of Hubei Province of China (No. 2021CFA079) and the Knowledge Innovation Program of Wuhan-Basic Research (No. 2022010801010319).


\begin{thebibliography}{99}

\bibitem{NA.TK2022} Anbar, N., Kalayci, T., Meidl, W., Riera, C., St$\check{a}$nic$\check{a}$, P.: P$c$N functions, complete mappings and quasi-group difference sets. arXiv:2212.12943 (2022). Available: https://arxiv.org/abs/2212.12943.

\bibitem{EB.AS1991} Biham, E., Shamir, A.: Differential cryptanalysis of DES-like cryptosystems. Journal of Cryptology \textbf{4}, 3--72 (1991). 

\bibitem{HD.PF2006} Dobbertin, H., Felke, P., Helleseth T., Rosendahl, P.: Niho type cross-correlation functions via dickson polynomials and Kloosterman sums. IEEE Transactions on Information Theory \textbf{5}(2), 613--627 (2006). 

\bibitem{PE.PF2020} Ellingsen, P., Felke, P., Riera, C., St$\check{a}$nic$\check{a}$, P., Tkachenko, A.: $C$-differentials, multiplicative uniformity, and (almost) perfect $c$-nonlinearity. IEEE Transactions on Information Theory \textbf{66}(9), 5781--5789 (2020). 

\bibitem{SH.MP2021} Hasan, S.U., Pal, M., Riera, C., St$\check{a}$nic$\check{a}$, P.: On the $c$-differential uniformity of certain maps over finite fields. Designs, Code and Cryptography \textbf{89}, 221--239 (2021). 

\bibitem{SM.CR2021} Mesnager, S., Riera, C., St$\check{a}$nic$\check{a}$, P., Yan, H., Zhou, Z.: Investigations on $c$-(almost) perfect nonlinear functions. IEEE Transactions on Information Theory \textbf{67}(10), 6916--6925 (2021). 

\bibitem{SM.MS2022} Mesnager S., Shi M., Zhu H.: Cyclic codes from low differentially uniform functions. arXiv:2210.12092 (2022). Available: https://arxiv.org/abs/2210.12092.

\bibitem{KN1994} Nyberg, K.: Differentially uniform mappings for cryptography. In: Helleseth, T. (eds) Advances in Cryptology - EUROCRYPT 1993, LNCS, vol. 765, pp. 55--64. Springer, Berlin, Heidelberg (1994). 

\bibitem{ZT.NL2023} Tu, Z., Li, N., Wu, Y., Zeng, X., Tang, X., Jiang, Y.: On the differential spectrum and the AP$c$N property of a class of power functions over finite fields. IEEE Transactions on Information Theory \textbf{69}(1), 582--597 (2023). 

\bibitem{XW.DZ2022} Wang, X., Zheng, D., Hu, L.: Several classes of P$c$N power functions over finite fields. Discrete Applied Mathematics \textbf{322}, 171--182 (2022). 

\bibitem{XX.SM2022} Xie, X., Mesnager, S., Li, N., He, D., Zeng, X.: On the Niho type locally-APN power functions and their boomerang spectrum. IEEE Transactions on Information Theory (2023). DOI: 10.1109/TIT.2022.3232362.

\bibitem{HY.KZ2022} Yan, H., Zhang, K.: On the $c$-differential spectrum of power functions over finite fields. Designs, Code and Cryptography \textbf{90}, 2385--2405 (2022). 

\end{thebibliography}
\end{document}